\documentclass[11pt]{article}
\usepackage{fullpage}
\usepackage{amsmath, amsthm, amssymb}
\usepackage{soul,color}
\usepackage[linesnumbered,ruled,vlined,boxed]{algorithm2e}
\usepackage{algorithmic}
\usepackage{epsfig}
\usepackage{graphicx}
\usepackage{graphics}
\usepackage{caption}
\usepackage{url}
\usepackage{verbatim}
\usepackage{enumerate}
\usepackage{subcaption}
\newtheorem{theorem}{Theorem}[section]
\newtheorem{lemma}[theorem]{Lemma}

\newtheorem{corollary}[theorem]{Corollary}

\title{An Improved Algorithm for Incremental DFS Tree in Undirected Graphs}
\date{}
\author{
Lijie Chen   \thanks{\texttt{lijieche@mit.edu}.} \\ Massachusetts Institute of Technology\and
Ran Duan \thanks{\texttt{duanran@mail.tsinghua.edu.cn}.}\\ Tsinghua University \and
Ruosong Wang \thanks{ \texttt{ruosongw@andrew.cmu.edu}.}\\ Carnegie Mellon University \and
Hanrui Zhang \thanks{\texttt{hrzhang@cs.duke.edu}.}\\ Duke University\and
Tianyi Zhang \thanks{\texttt{tianyi-z16@mails.tsinghua.edu.cn}.}\\ Tsinghua University
}
\begin{document}
\newcommand{\dfs}{\textsf{DFS}}
\newcommand{\pathq}{\textsf{PathQuery}}
\newcommand{\casedge}{\mathcal{F}}
\newcommand{\casanc}{\textsf{anc}}
\newcommand{\cas}{\textsf{cas}}
\newcommand{\size}{\textsf{size}}
\newcommand{\incremain}{\textsf{BatchInsert}}
\newcommand{\data}{\mathcal{D}}
\maketitle

\begin{abstract}
	Depth first search (DFS) tree is one of the most well-known data structures for designing efficient graph algorithms. Given an undirected graph $G=(V,E)$ with $n$ vertices and $m$ edges, the textbook algorithm takes $O(n+m)$ time to construct a DFS tree. In this paper, we study the problem of maintaining a DFS tree when the graph is undergoing incremental updates. Formally, we show:
	\begin{quote}
		Given an arbitrary online sequence of edge or vertex insertions, there is an algorithm that reports a DFS tree in $O(n)$ worst case time per operation, and requires $O\left(\min\{m \log n, n^2\}\right)$ preprocessing time.
	\end{quote}
	
	Our result improves the previous $O(n \log^3 n)$ worst case update time algorithm by Baswana et al.~\cite{baswana2016dynamic} and the $O(n \log n)$ time by Nakamura and Sadakane~\cite{nakamura2017space}, and matches the trivial $\Omega(n)$ lower bound when it is required to explicitly output a DFS tree.
	
	Our result builds on the framework introduced in the breakthrough work by Baswana et al.~\cite{baswana2016dynamic}, together with a novel use of a tree-partition lemma by Duan and Zhang~\cite{duan2016improved}, and the celebrated fractional cascading technique by Chazelle and Guibas~\cite{chazelle1986fractional,chazelle1986fractional2}.
\end{abstract}

\section{Introduction}

Depth First Search (DFS) is one of the most renowned graph traversal techniques. After Tarjan's seminal work~\cite{tarjan1972depth}, it demonstrates its power by leading to efficient algorithms to many fundamental graph problems, e.g., biconnected components, strongly connected components, topological sorting, bipartite matching, dominators in directed graph and planarity testing.

Real world applications often deal with graphs that keep changing with time. Therefore it is natural to study the dynamic version of graph problems, where there is an online sequence of updates on the graph, and the algorithm aims to maintain the solution of the studied graph problem efficiently after seeing each update.
The last two decades have witnessed a surge of research in this area, like connectivity~\cite{eppstein1997sparsification,henzinger1999randomized,holm2001poly,kapron2013dynamic}, reachability~\cite{roditty2008improved,sankowski2004dynamic}, shortest path~\cite{demetrescu2004new,roditty2012dynamic}, bipartite matching~\cite{baswana2011fully,neiman2016simple}, and min-cut~\cite{thorup2001fully}. 


We consider the dynamic maintenance of DFS trees in undirected graphs. As observed by Baswana et al.~\cite{baswana2016dynamic} and Nakamura and Sadakane~\cite{nakamura2017space}, the {\em incremental} setting, where edges/vertices are added but never deleted from the graph, is arguably easier than the {\em fully dynamic} setting where both kinds of updates can happen --- in fact, they provide algorithms for incremental DFS with $\tilde{O}(n)$ worst case update time, which is close to the trivial $\Omega(n)$ lower bound when it is required to explicitly report a DFS tree after each update. {\bf\em So, is there an algorithm that requires nearly linear preprocessing time and space, and reports a DFS tree after each incremental update in $O(n)$ time?} In this paper, we study the problem of maintaining a DFS tree in the incremental setting, and give an affirmative answer to this question.

\subsection{Previous works on dynamic DFS}

Despite the significant role of DFS tree in static algorithms, there is limited progress on maintaining a DFS tree in the {\em dynamic} setting.

Many previous works focus on the {\em total time} of the algorithm for any arbitrary updates. Franciosa et al.~\cite{franciosa1997incremental} designed an incremental algorithm for
maintaining a DFS tree in a DAG from a given source, with $O(mn)$ total time for an arbitrary sequence of edge insertions; 
Baswana and Choudhary~\cite{baswana2015dynamic} designed a decremental algorithm for maintaining a DFS tree in a DAG with expected $O(mn\log n)$ total time.
For undirected graphs, Baswana and Khan~\cite{baswana2014incremental} designed an incremental algorithm for maintaining a DFS tree with $O(n^2)$ total time.

These algorithms used to be the only results known for the dynamic DFS tree problem. 
However, none of these existing algorithms, despite that they are designed for only a partially dynamic environment, 
achieves a worst case bound of $o(m)$ on the update time. 

That barrier is overcome in the recent breakthrough work of Baswana et al.~\cite{baswana2016dynamic}, they provide, for undirected graphs, a fully dynamic algorithm with worst case $O(\sqrt{mn} \log^{2.5} n)$ update time, and an incremental algorithm with worst case $O(n \log^{3} n)$ update time. Due to the rich information in a DFS tree, their results directly imply faster worst case fully dynamic algorithms for subgraph connectivity, biconnectivity and 2-edge connectivity.

The results of Baswana et al.~\cite{baswana2016dynamic} suggest a promising way to further improve the worst case update time or space consumption for those fully dynamic algorithms by designing better dynamic algorithms for maintaining a DFS tree. In particular, based on the framework by Baswana et al.~\cite{baswana2016dynamic}, Nakamura and Sadakane~\cite{nakamura2017space} propose an algorithm which takes $O(\sqrt{mn} \log^{1.75} n / \sqrt{\log \log n})$ time per update in the fully dynamic setting and $O(n \log n)$ time in the incremental setting, and $O(m \log n)$ bits of space.

\subsection{Our results}
	
	In this paper, following the approach of~\cite{baswana2016dynamic}, we improve the update time for the incremental setting, also studied in~\cite{baswana2016dynamic}, by combining a better data structure, a novel tree-partition lemma by Duan and Zhang~\cite{duan2016improved} and the fractional-cascading technique by Chazelle and Guibas~\cite{chazelle1986fractional,chazelle1986fractional2}.
	
	For any set $U$ of incremental updates (insertion of a vertex/an edge), we let $G + U$ denote the graph obtained by applying the updates in $U$ to the graph $G$. Our results build on the following main theorem.
	
	\begin{theorem}\label{batch-ins}
		There is a data structure with $O(\min\{m \log n, n^2\})$ size, and can be built in $O(\min\{m \log n, n^2\})$ time, such that given a set $U$ of $k$ insertions, a DFS tree of $G + U$ can be reported in $O(n+k)$ time.
	\end{theorem}
	
	By the above theorem combined with a de-amortization trick in~\cite{baswana2016dynamic}, we establish the following corollary for maintaining a DFS tree in an undirected graph with incremental updates.
	
		
		
	\begin{corollary}[\textbf{Incremental DFS tree}]\label{cor-incre-dfs}
		Given a sequence of online edge/vertex insertions, a DFS tree can be maintained in $O(n)$ worst case time per insertion.
	\end{corollary}
	

\subsection{Organization of the Paper}

In Section~2 we introduce frequently used notations and review two building blocks of our algorithm --- the tree partition structure \cite{duan2016improved} and the fractional cascading technique \cite{chazelle1986fractional,chazelle1986fractional2}.
In Section~3, we consider a batched version of the incremental setting, where all incremental updates are given at once, after which a single DFS tree is to be reported. Given an efficient scheme to answer queries of form $Q(T(\cdot), \cdot, \cdot)$, we prove Theorem~\ref{batch-ins}, which essentially says there is an efficient algorithm, which we call \incremain{}, for the batched incremental setting.
In Section~4, we elaborate on the implementation of the central query subroutine $Q(T(\cdot), \cdot, \cdot)$ used in the batched incremental algorithm.
We first review a standard de-amortization technique, applying which our algorithm for the batched setting directly implies the efficient algorithm for the incremental setting stated in Corollary~\ref{cor-incre-dfs}.
We then, in Sections~\ref{sec:logn}~and~\ref{sec:nsquare} respectively, introduce (1) an optimized data structure that takes $O(m \log n)$ time for preprocessing and answers each query in $O(\log n)$ time, and (2) a relatively simple data structure that takes $O(n^2)$ time for preprocessing and answers each query in $O(1)$ time.
One of these two structures, depending on whether $m \log n > n^2$ or not, is then used in Section~\ref{sec:mlogn} to implement a scheme that answers each query in amortized $O(1)$ time.
This is straightforward when the $(n^2, 1)$ structure is used.
When instead the $(m \log n, \log n)$ structure is used, we apply a nontrivial combination of the tree partition structure and the fractional cascading technique to bundle queries together, and answer each bundle using a single call to the $(m \log n, \log n)$ structure.
We show that the number of such bundles from queries made by \incremain{} cannot exceed $O(n / \log n)$, so the total time needed for queries is $O(n)$.
This finishes the proof of Theorem~\ref{batch-ins} and Corollary~\ref{cor-incre-dfs} and concludes the paper.

\section{Preliminaries}

Let $G = (V, E)$ denote the original graph, $T$ a corresponding DFS tree, and $U$ a set of inserted vertices and edges. We first introduce necessary notations.
\begin{itemize}
    \item $T(x)$: The subtree of $T$ rooted at $x$.
    \item $\mathit{path}(x, y)$: The path from $x$ to $y$ in $T$.
    \item $\mathit{par}(v)$: The parent of $v$ in $T$.
    \item $N(x)$: The adjacency list of $x$ in $G$.
    \item $L(x)$: The reduced adjacency list for vertex $x$, which is maintained during the algorithm.
    \item $T^*$: The newly generated DFS tree.
    \item $\mathit{par}^*(v)$: The parent of $v$ in $T^*$.
\end{itemize}

Our result uses a tree partition lemma in \cite{duan2016improved} and the famous fractional cascading structure in \cite{chazelle1986fractional,chazelle1986fractional2}, which are summarized as the following two lemmas. 

\begin{lemma}[Tree partition structure \cite{duan2016improved}]
\label{tree_partition_lem}
Given a rooted tree $T$ and any integer parameter $k$ such that $2 \le k \le n = |V(T)|$, there exists a subset of vertices $M \subseteq V(T)$, $|M|\le 3n/k-5$, such that after removing all vertices in $M$, the tree $T$ is partitioned into sub-trees of size at most $k$. We call every $v \in M$ an $M$-marked vertex, and $M$ a marked set. Also, such $M$ can be computed in $O(n \log n)$ time.
\end{lemma}

\begin{lemma}[Fractional cascading \cite{chazelle1986fractional,chazelle1986fractional2}]
\label{fractional_cascading}
Given $k$ sorted arrays $\{A_i\}_{i \in [k]}$ of integers with total size $\sum_{i=1}^k |A_i| = m$.
There exists a data structure which can be built in $O(m)$ time and using $O(m)$ space,
such that for any integer $x$,
the successors of $x$ in all $A_i$'s can be found in $O(k + \log m)$ time.
\end{lemma}

\section{Handling batch insertions}

        In this section, we study the dynamic DFS tree problem in the batch insertion setting.
        The goal of this section is to prove Theorem \ref{batch-ins}.
        Our algorithm basically follows the same framework for fully dynamic DFS proposed in \cite{baswana2016dynamic}.
        Since we are only interested in the dynamic DFS tree problem in the batch insertion setting, the algorithms \textsf{BatchInsert} and \textsf{DFS} presented below is a moderate simplification of the original algorithm in \cite{baswana2016dynamic}, by directly pruning those details unrelated to insertions.

        \begin{algorithm}[H]
            \caption{\incremain}
            \KwData{a DFS tree $T$ of $G$, set of insertions $U$}
            \KwResult{a DFS tree $T^*$ of $G + U$}
            Add each inserted vertex $v$ into $T$, set $\mathit{par}(v) = r$\;
            Initialize $L(v)$ to be $\emptyset$ for each $v$\;
            Add each inserted edge $(u, v)$ to $L(u)$ and $L(v)$\;
            Call $\mathrm{\dfs}(r)$\;
        \end{algorithm}
        \begin{algorithm}[H]
            \caption{\dfs}
            \KwData{a DFS tree $T$ of $G$, the entering vertex $v$}
            \KwResult{a partial DFS tree}
            Let $u = v$\;
            \While{$\mathit{par}(u)$ is not visited} {
                Let $u = \mathit{par}(u)$\;
            }
            Mark $\mathit{path}(u, v)$ to be visited\;
            Let $(w_1, \dots, w_t) = \mathit{path}(u, v)$\;
            \For{$i \in [t]$} {
                \If{$i \ne t$} {
                    Let $\mathit{par}^*(w_i) = w_{i + 1}$\;
                }
                \For{child $x$ of $w_i$ in $T$ except $w_{i + 1}$} {
                    Let $(y, z) = Q(T(x), u, v)$, where $y \in \mathit{path}(u, v)$\;
                    Add $z$ into $L(y)$\;
                }
            }
            \For{$i \in [t]$} {
                \For{$x \in L(w_i)$} {
                    \If{$x$ is not visited} {
                        Let $\mathit{par}^*(x) = w_i$\;
                        Call $\mathrm{\dfs}(x)$\;
                    }
                }
            }
        \end{algorithm}

        In Algorithm \incremain, we first attach each inserted vertex to the super root $r$, and pretend it has been there since the very beginning. Then only edge insertions are to be considered. All inserted edges are added into the reduced adjacency lists of corresponding vertices. We then use \dfs{}  to traverse the graph starting from $r$ based on $T$, $L$, and build the new DFS tree while traversing the entire graph and updating the reduced adjacency lists.

        In Algorithm \textsf{DFS}, the new DFS tree is built in a recursive fashion. Every time we enter an untouched subtree, say $T(u)$, from vertex $v \in T(u)$, we change the root of $T(u)$ to $v$ and go through $\mathit{path}(v, u)$; i.e., we wish to reverse the order of $\mathit{path}(u, v)$ in $T^*$. One crucial step behind this operation is that we need to find a new root for each subtree $T(w)$ originally hanging on $\mathit{path}(u, v)$. The following lemma tells us where the $T(w)$ should be rerooted on $\mathit{path}(u, v)$ in $T^*$.
                
        \begin{lemma}[\cite{baswana2016dynamic}]
        	\label{feasible_edge}
        	Let $T^*$ be a partially constructed DFS tree, $v$ the current vertex being visited, $w$ an (not necessarily proper) ancestor of $v$ in tree $T^*$, and $C$ a connected component of the subgraph induced by unvisited vertices. If there are two edges $e$ and $e'$ from $C$ incident on $v$ and $w$, then it is sufficient to consider only $e$ during the rest of the DFS traversal.       
        \end{lemma}        
        
        Let $Q(T(w), u, v)$ be the edge between the highest vertex on $\mathit{path}(u, v)$ incident to a vertex in subtree $T(w)$, and the corresponding vertex in $T(w)$.
        $Q(T(w), u, v)$ is defined to be $\mathsf{Null}$ if such an edge does not exist.
        By Lemma \ref{feasible_edge}, it suffices to ignore all other edges but just keep the edge returned by $Q(T(w), u, v)$; this is because we have reversed the order of $\mathit{path}(u, v)$ in $T^*$ and thus $Q(T(w), u, v)$ connects to the lowest possible position in $T^*$. Hence $T(w)$ should be rerooted at $Q(T(w), u, v)$.
        
        Denote $(x, y)$ to be the edge returned by $Q(T(w), u, v)$ where $x \in \mathit{path}(u, v)$, and then we add $y$ into $L(x)$. After finding an appropriate entering edge for each hanging subtree, we process each vertex $v \in \mathit{path}(u, v)$ in ascending order of depth (with respect to tree $T$). For every unvisited $w \in L(v)$, we set $\mathit{par}^*(w) = v$, and recursively call $\mathrm{\dfs}(w)$.

        \begin{theorem}
            \incremain{} correctly reports a feasible DFS tree $T^*$ of graph $G + U$.
        \end{theorem}
        \begin{proof}
            We argue that in a single call $\mathrm{\dfs}(v)$, where $u$ is the highest unvisited ancestor of $v$, every unvisited (at the moment of being enumerated) subtree $T(w)$ hanging from $\mathit{path}(u, v)$, as well as every vertex on $\mathit{path}(u, v)$ except $v$, will be assigned an appropriate parent such that these parent-child relationships constitute a DFS tree of $G$ at the termination of \incremain{}. When the traversal reaches $v$, the entire $T(u)$ is untouched, or else $u$ would have been marked by a previous visit to some vertex in $T(u)$. We could therefore choose to go through $\mathit{path}(v, u)$ to reach $u$ first. By Lemma~\ref{feasible_edge}, if a subtree $T(w)$ is reached from some vertex on $\mathit{path}(u, v)$, it suffices to consider only the edge $Q(T(w), u, v)$. After adding the query results of all hanging subtrees into the adjacency lists of vertices on $\mathit{path}(u, v)$, every hanging subtree visited from some vertex $x$ on $\mathit{path}(u, v)$ should be visited in a correct way through edges in $L(x)$ solely. Since every vertex will eventually be assigned a parent, \incremain{} does report a feasible DFS tree of graph $G + U$.
        \end{proof}
        
	For now we have not discussed how to implement $Q(T(w), u, v)$ and the above algorithm only assumes blackbox queries to $Q(T(\cdot), \cdot, \cdot)$. The remaining problem is to devise a data structure $\mathcal{D}$ to answer all the queries demanded by Algorithm \textsf{DFS} in $O(n)$ total time. We will show in the next section that there exists a data structure $\mathcal{D}$ with the desired performance, which is stated as the following lemma.
        \begin{lemma}
        \label{query_time}
            There exists a data structure $\mathcal{D}$ with preprocessing time $O\left(\min\{m \log n, n^2\}\right)$ time and space complexity $O\left( \min \{ m \log n, n^2 \} \right)$ that can answer all queries $Q(T(w), x, y)$  in a single run of \incremain{} in $O(n)$ time. 
        \end{lemma}

        \begin{proof}[Proof of Theorem \ref{batch-ins}]
            By Lemma~\ref{query_time}, the total time required to answer queries is $O(n)$. The total size of reduced adjacency lists is bounded by $O(n + |U|)$, composed by $O(|U|)$ edges added in \incremain{} and $O(n)$ added during DFS. Thus, the total time complexity of $\incremain$ is $O(n + |U|)$.

            During preprocessing, we use depth first search on $G$ to get the initial DFS tree $T$, and build $\data$ in time $O\left( \min \{ m \log n, n^2 \} \right)$. The total time for preprocessing is $O\left( \min \{ m \log n, n^2 \} \right)$.
        \end{proof}

     \section{Dealing with queries in \incremain}
    In this section we prove Lemma~\ref{query_time}. Once this goal is achieved, the overall time complexity of batch insertion taken by Algorithm~\incremain{} would be $O(n + |U|)$.
    
	In the following part of this section, we will first devise a data structure in Section \ref{sec:logn}, that answers any single query $Q(T(w), u, v)$ in $O(\log n)$ time, which would be useful in other parts of the algorithm. We will then present another simple data structure in Section \ref{sec:nsquare}, which requires $O(n^2)$ preprocessing time and $O(n^2)$ space and answers each query in $O(1)$ time. Finally, we propose a more sophisticated data structure in Section \ref{sec:mlogn}, which requires $O(m\log n)$ preprocessing time and $O(m \log n)$ space and answer all queries $Q(T(w), x, y)$  in a single run of \incremain{} in $O(n)$ time. Hence, we can always have an algorithm that handles a batch insertion $U$ in $O(n + |U|)$ time using $O(\min\{m\log n, n^2\})$ preprocessing time and  $O(\min\{m\log n, n^2\})$ space, thus proving Theorem \ref{batch-ins}. We can then prove Corollary \ref{cor-incre-dfs} using the following standard de-amortization argument. 
        \begin{lemma}{(Lemma 6.1 in \cite{baswana2016dynamic})}
        	\label{deamortization}
        	Let $\data$ be a data structure that can be used to report the solution of a graph problem after a set of $U$ updates on an input graph $G$. If $\data$ can be initialized in $O(f)$ time and the solution for graph $G + U$ can be reported in $O(h + |U| \times g)$ time, then $\data$ can be modified to report the solution after every update in worst-case $O\left(\sqrt{fg} + h\right)$ update time after spending $O(f)$ time in initialization, given that $\sqrt{f / g} \le n$.
        \end{lemma}
    
	\begin{proof}[Proof of Corollary \ref{cor-incre-dfs}]
		Taking $f = \min\{m\log n, n^2 \}$, $g = 1$, $h = n$ and directly applying the above lemma will yield the desired result. 
	\end{proof}
    
	\subsection{Answering a single query in $O(\log n)$ time} \label{sec:logn}
	 We show in this subsection that the query $Q(T(\cdot), \cdot, \cdot)$ can be reduced efficiently to the range successor query (see, e.g., \cite{nekrich2012sorted}, for the definition of range successor query), and show how to answer the range successor query, and thus any individual query $Q(T(\cdot), \cdot, \cdot)$, in $O(\log n)$ time.

	To deal with a query $Q(T(w), x, y)$, first note that since $T$ is a DFS tree, all edges not in $T$ but in the original graph $G$ must be ancestor-descendant edges. Querying edges between $T(w)$ and $\mathit{path}(x, y)$ where $x$ is an ancestor of $y$ and $T(w)$ is hanging from $\mathit{path}(x, y)$ is therefore equivalent to querying edges between $T(w)$ and $\mathit{path}(x, \mathit{par}(w))$, i.e., $Q(T(w), x, y) = Q(T(w), x, \mathit{par}(w))$. From now on, we will consider queries of the latter form only.
	
	Consider the DFS sequence of $T$, where the $i$-th element is the $i$-th vertex reached during the DFS on $T$. Note that every subtree $T(w)$ corresponds to an interval in the DFS sequence. Denote the index of vertex $v$ in the DFS sequence by $\mathit{first}(v)$, and the index of the last vertex in $T(v)$ by $\mathit{last}(v)$. 
	During the preprocessing, we build a 2D point set $S$. For each edge $(u, v) \in E$, we add a point  $p = (\mathit{first}(u), \mathit{first}(v))$ into $S$. Notice that for each point $p \in S$, there exists exactly one edge $(u, v)$ associated with $p$. Finally we build a 2D range tree on point set $S$ with $O(m\log n)$ space and $O(m\log n)$ preprocessing time.
	
	To answer an arbitrary query $Q(T(w), x, \mathit{par}(w))$, we query the point with minimum $x$-coordinate lying in the rectangle $\Omega = [\mathit{first}(x), \mathit{first}(w) - 1] \times [\mathit{first}(w), \mathit{last}(w)]$. If no such point exists, we return \textsf{Null} for $Q(T(w), x, \mathit{par}(w))$. Otherwise we return the edge corresponding to the point with minimum $x$-coordinate.
	
	Now we prove the correctness of our approach.
	\begin{itemize}
		\item If our method returns \textsf{Null}, $Q(T(w), x, \mathit{par}(w))$ must equal \textsf{Null}. Otherwise, suppose $Q(T(w), x, \mathit{par}(w)) = (u, v)$. Noticing that $(\mathit{first}(u), \mathit{first}(v))$ is in $\Omega$, it means our method will not return \textsf{Null} in that case.
		
		\item If our method does not return \textsf{Null}, denote $(u', v')$ to be the edge returned by our method. We can deduce from the query rectangle that $u' \in T(x) \backslash T(w)$ and $v' \in T(w)$. Thus, $Q(T(w), x, \mathit{par}(w)) \neq \textsf{Null}$. Suppose $Q(T(w), x, \mathit{par}(w)) = (u, v)$. Notice that $(\mathit{first}(u), \mathit{first}(v))$ is in $\Omega$, which means $\mathit{first}(u') \le \mathit{first}(u)$. If $u' = u$, then our method returns a feasible solution. Otherwise, from the fact that $\mathit{first}(u') < \mathit{first}(u)$, we know that $u'$ is an ancestor of $u$, which contradicts the definition of  $Q(T(w), x, \mathit{par}(w))$.
	\end{itemize}    
    
	\subsection{An $O(n^2)$-space data structure}\label{sec:nsquare}
	\label{pre_all}
	
	In this subsection we propose a data structure with quadratic preprocessing time and space complexity that answers any $Q(T(\cdot), \cdot, \cdot)$ in constant time.
	
	Since we allow quadratic space, it suffices to precompute and store answers to all possible queries $Q(T(w), u, \mathit{par}(w))$. For preprocessing, we enumerate each subtree $T(w)$, and fix the lower end of the path to be $v = \mathit{par}(w)$ while we let the upper end $u$ go upward from $v$ by one vertex at a time to calculate $Q(T(w), u, v)$ incrementally, in order to get of the form $Q(T(w), \cdot, \cdot)$ in $O(n)$ total time.
	
	As $u$ goes up, we check whether there is an edge from $T(w)$ to the new upper end $u$ in $O(1)$ time; for this task we build an array (based on the DFS sequence of $T$) for each vertex, and insert an 1 into the appropriate array for each edge, and apply the standard prefix summation trick to check whether there is an 1 in the range corresponding to $T(w)$. To be precise, let $A_u: [n] \rightarrow \{0, 1\}$ denote the array for vertex $u$. Recall that $\mathit{first}(v)$ denotes the index of vertex $v$ in the DFS sequence, and $\mathit{last}(v)$ the index of the last vertex in $T(v)$. For a vertex $u$, we set $A_u[\mathit{first}(v)]$ to be 1 if and only if there is an edge $(u, v)$ where $u$ is the higher end. Now say, we have the answer to $Q(T(w), u, v)$ already, and want to get $Q(T(w), u', v)$ in $O(1)$ time, where $u' = \mathit{par}(u)$. If there is an edge between $T(w)$ and $u'$, then it will be the answer. Or else the answer to $Q(T(w), u', v)$ will be the same as to $Q(T(w), u, v)$. In order to know whether there is an edge between $T(w)$ and $u'$, we check the range $[\mathit{first}(w), \mathit{last}(w)]$ in $A_{u'}$, and see if there is an $1$ in $O(1)$ time using the prefix summation trick. 
	
	\begin{lemma}
		\label{preprocess_all}
		The preprocessing time and query time of the above data structure are $O(n^2)$ and $O(1)$ respectively.
	\end{lemma}
	\begin{proof}
		The array $A_u$ and its prefix sum can be computed for each vertex $u$ in total time $O(n^2)$. For each subtree $T(w)$, we go up the path from $w$ to the root $r$, and spend $O(1)$ time for each vertex $u$ on $\mathit{path}(r, w)$ to get the answer for $Q(T(w), u, \mathit{par}(w))$. There are at most $n$ vertices on $\mathit{path}(r, w)$, so the time needed for a single subtree is $O(n)$, and that needed for all subtrees is $n \cdot O(n) = O(n^2)$ in total. On the other hand, for each query, we simply look it up and answer in $O(1)$ time. Hence we conclude that the preprocessing time and query time are $O(n^2)$ and $O(1)$ respectively.
	\end{proof}    
   
   \subsection{An $O(m\log n)$-space data structure}\label{sec:mlogn}
		Observe that in \incremain{} (and \dfs), a bunch of queries $\{Q(T(w_i), x, y)\}$ are always made simultaneously, where $\{T(w_i)\}$ is the set of subtrees hanging from $\mathit{path}(x, y)$. We may therefore answer all queries for a path in one pass, instead of answering them one by one. By doing so we confront two types of hard queries.

        First consider an example where the original DFS tree $T$ is a chain $L$ where $a_1$ is the root of $L$ and for $1 \le i \le n - 1$, $a_{i+1}$ is the unique child of $a_i$. When we invoke $\dfs(a_1)$ on $L$, $path(u, v)$ is the single node $a_1$. Thus, we will call $Q(T(a_2), a_1, a_1)$ and add the returned edge into $L(a_1)$. Supposing there are no back-edges in this graph, the answer of $Q(T(a_2), a_1, a_1)$ will be the edge $(a_1, a_2)$. Therefore, we will recursively call the $\dfs(a_2)$ on the chain $(a_2, a_n)$. Following further steps of \dfs, we can see that we will call the query $Q(T(w), x, y)$ for $\Omega(n)$ times. For the rest of this subsection, we will show that we can deal with this example in linear time. The idea is to answer queries involving short paths in constant time. For instance, in the example shown above, $path(u, v)$ always has constant length. We show that when the length of $path(u, v)$ is smaller than $2\log n$, it is affordable to preprocess all the answers to queries of this kind in $O(m\log n)$ time and $O(n \log n)$ space.
        
	    \begin{figure}[!ht]
			\centering
		    \includegraphics[width=0.6\linewidth]{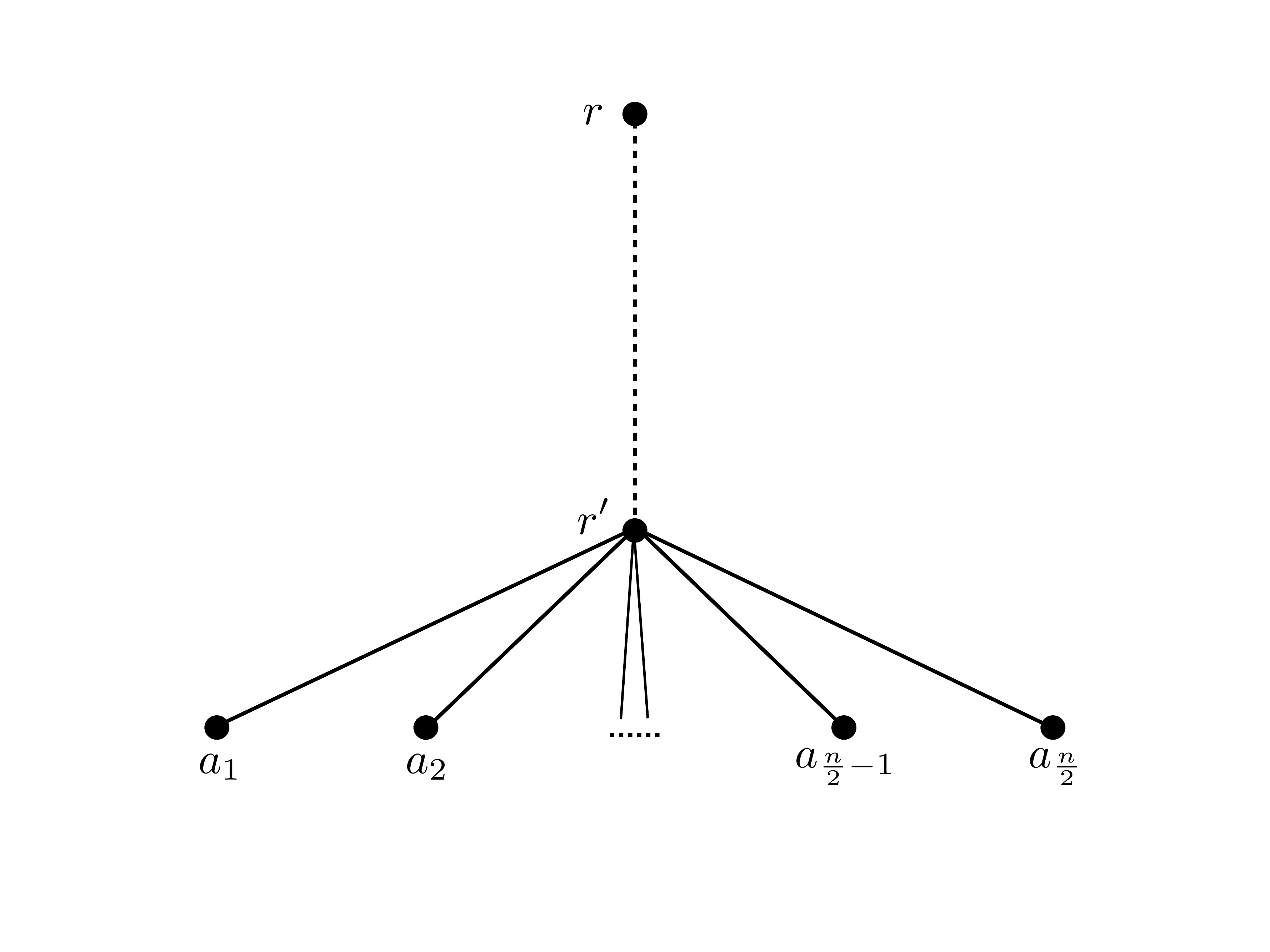} 
			\captionsetup{justification=centering}
			\caption{In this example, if we stick to the 2D-range-based data structure introduced before, then computing all $Q(T(a_i), r, r^\prime)$ would take as much as $O(n\log n)$ time.}
			\label{fig:heavy}
		\end{figure}
		The second example we considered is given as Figure \ref{fig:heavy}. In this tree, the original root is $r$. Suppose the distance between $r$ and $r'$ is $n / 2$. When we invoke $\dfs(r')$, $path(u, v)$ the path from $r$ to $r'$. Thus, we will call $T(a_1, r, r')$, $T(a_2, r, r')$, $\ldots$, $T(a_{n - 2}, r, r')$, which means we make $\Omega(n)$ queries. 
		In order to deal with this example in linear time, the main idea is using fractional cascading to answer all queries $Q(T(w), x, y)$ with a fixed $path(u, v)$, for all subtrees $T(w)$ with small size.
		
		In the examples shown above, all subtrees cut off $path(u, v)$ have constant size and thus the total time complexity for this example is $O(n)$. We will finally show that, by combining the two techniques mentioned above, it is enough to answer all queries $Q(T(w), x, y)$ in linear time, thus proving Lemma~\ref{query_time}.

    
    \subsection*{Data structure}
	The data structure consists of the following parts.
	\begin{enumerate}[(\romannumeral1)]
		\item Build the 2D-range successor data structure that answers any $Q(T(\cdot), \cdot, \cdot)$ in $O(\log n)$ time.
		\item For each ancestor-descendent pair $(u, v)$ such that $u$ is at most $2\log n$ hops above $v$, precompute and store the value of $Q(T(v), u, \mathit{par}(v))$.
		\item Apply Lemma \ref{tree_partition_lem} with parameter $k = \log n$ and obtain a marked set of size $O(n / \log n)$. Let $M$ be the set of all marked vertices $x$ such that $|T(x)| \geq \log n$. For every $v\notin M$, let $\mathit{anc}_v\in M$ be the nearest ancestor of $v$ in set $M$.
		
		Next we build a fractional cascading data structure for each $u\in M$ in the following way. Let $M_u$ be the set of all vertices in $T(u)$ whose tree paths to $u$ do not intersect any other vertices $u^\prime \neq u$ from $M$, namely $M_u = \{v\mid \mathit{anc}_v = u \}$; see Figure \ref{ds} for an example. Then, apply Lemma \ref{fractional_cascading} on all $N(v), v\in M_u$ where $N(v)$ is treated as sorted array in an ascending order with respect to depth of the edge endpoint opposite to $v$; this would build a fractional cascading data structure that, for any query encoded as a $w\in V$, answers for every $v\in M_u$ its highest neighbour below vertex $w$ in total time $O(|M_u| + \log n)$.
	\end{enumerate}
	
	Here is a structural property of $M$ that will be used when answering queries.
	\begin{lemma}\label{struct-marked}
		For any ancestor-descendent pair $(u, v)$, if $\mathit{path}(u, v)\cap M = \emptyset$, then $\mathit{path}(u, v)$ has $\leq 2\log n$ hops.
	\end{lemma}
	\begin{proof}
		Suppose otherwise. By definition of marked vertices there exists a marked vertex $w\in \mathit{path}(u, v)$ that is $\leq \log n$ hops below $u$. Then since $\mathit{path}(u, v)$ has $>2\log n$ many hops, it must be $T(w)\geq \log n$ which leads to $w\in M$, contradicting $\mathit{path}(u, v)\cap M = \emptyset$.
	\end{proof}
	
		\begin{center}
			\begin{figure}
				\begin{subfigure}{0.5\textwidth}
					\includegraphics[width=3in, height=2.4in]{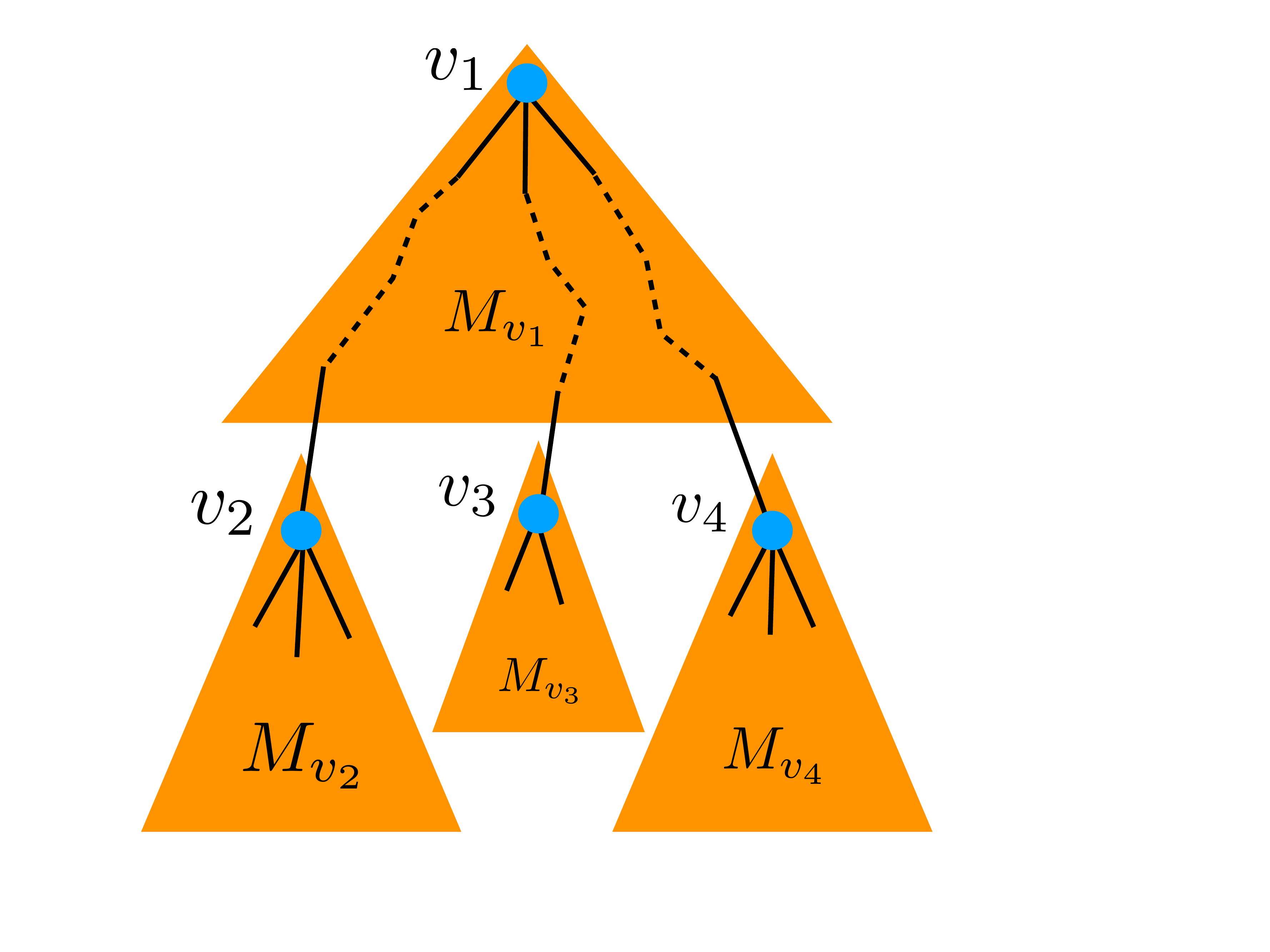}
					\caption{In this example, each blue node represents a vertex $v_i (1\leq i\leq 4)$ from set $M$, and $M_{v_i}$'s are drawn as yellow triangles. For each triangle, a fractional cascading data structure is built on adjacency lists of all vertices inside.}
					\label{ds}
				\end{subfigure}
				\quad
				\begin{subfigure}{0.5\textwidth}
					\includegraphics[width=3.3in, height=2.5in]{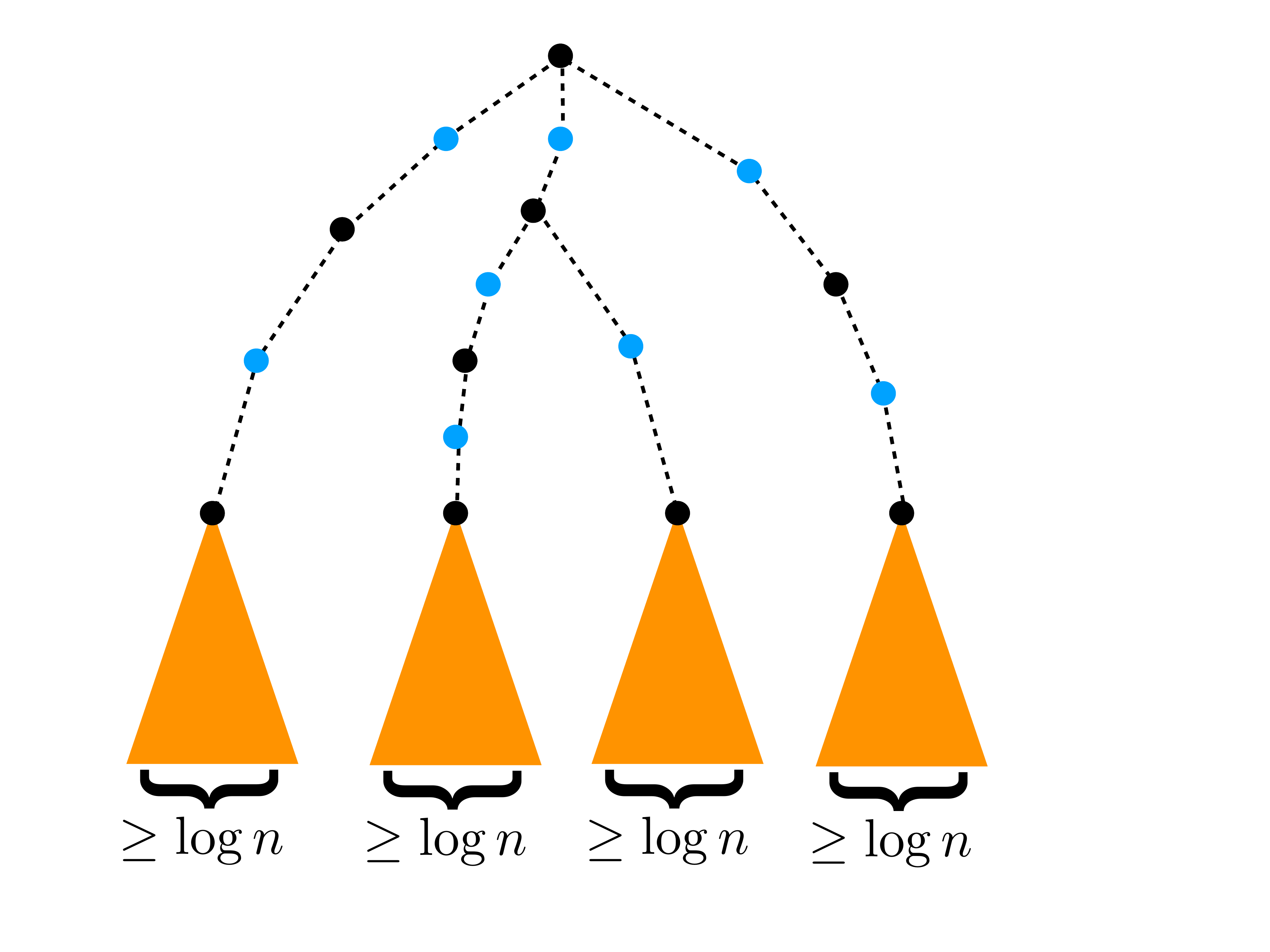}
					\caption{In this picture, sets $M$ and $X\cup \{r\}$ are drawn as blue nodes and black nodes respectively, and each yellow triangle is a subtree rooted at a leaf of $T[X]$, which has size $\geq \log n$. Note that every ancestor-descendent tree path between two black nodes contains a blue node.}
					\label{qry}
				\end{subfigure}
			\end{figure}
		\end{center}
	
	\subsection*{Preprocessing time}
	First of all, for part (\romannumeral1), as discussed in a previous subsection, 2D-range successor data structure takes time $O(m\log n)$ to initialize. Secondly, for part (\romannumeral3), on the one hand by Lemma \ref{tree_partition_lem} computing a tree partition takes time $O(n\log n)$; on the other hand, by Lemma \ref{fractional_cascading}, initializing the fractional cascading with respect to $u\in M$ costs $O(\sum_{v\in M_u}|N(v)|)$ time. Since, by definition of $M_u$, each $v\in V$ is contained in at most one $M_u, u\in M$, the overall time induced by this part would be $O(\sum_{u\in M}\sum_{v\in M_u}|N(v)|) = O(m)$.
	
	Preprocessing part (\romannumeral2) requires a bit of cautions. The procedure consists of two steps.
	\begin{enumerate}[(1)]
		\item For every ancestor-descendent pair $(u, v)$ such that $u$ is at most $2\log n$ hops above $v$, we mark $(u, v)$ if $u$ is incident to $T(v)$.
		
		Here goes the algorithm: for every edge $(u, w)\in E$ ($u$ being the ancestor), let $z\in \mathit{path}(u, w)$ be the vertex which is $2\log n$ hops below $u$ (if $\mathit{path}(u, w)$ has less than $2\log n$ hops, then simply let $z = w$); note that this $z$ can be found in constant time using the level-ancestor data structure \cite{bender2000lca} which can be initialized in $O(n)$ time. Then, for every vertex $v\in \mathit{path}(u, z)$, we mark the pair $(u, v)$. The total running time of this procedure is $O(m\log n)$ since each edge $(u, w)$ takes up $O(\log n)$ time.
		
		\item Next, for each $v\in V$, we compute all entries $Q(T(v), u, \mathit{par}(v))$ required by (\romannumeral2) in an incremental manner. Let $u_1, u_2, \cdots, u_{2\log n}$ be the nearest $2\log n$ ancestors of $v$ sorted in descending order with respect to depth, and then we directly solve the recursion 
		$Q(T(v), u_{i+1}, \mathit{par}(v)) = \begin{cases} 
		Q(T(v), u_{i}, \mathit{par}(v))	&	(u_{i+1}, v)\text{ is not marked}\\
		u_{i+1}	&	i=0\text{ or }(u_{i+1}, v)\text{ is marked}\\
		\end{cases}$ for all $0\leq i < 2\log n$ in $O(\log n)$ time. The total running time would thus be $O(n\log n)$.
	\end{enumerate}
	
	Summing up (\romannumeral1)(\romannumeral2)(\romannumeral3), the preprocessing time  is bounded by $O(m\log n)$.

	\subsection*{Query algorithm and total running time}
	We show how to utilize the above data structures (\romannumeral1)(\romannumeral2)(\romannumeral3) to implement $Q(T(\cdot), \cdot, \cdot)$ on line 9-11 in Algorithm \textsf{DFS} such that the overall time complexity induced by this part throughout a single execution of Algorithm \incremain{} is bounded by $O(n)$.
	
	Let us say we are given $(w_1, w_2, \cdots, w_t) = \mathit{path}(u, v)$ and we need to compute $Q(T(x), u, v)$ for every subtree $T(x)$ that is hanging on $\mathit{path}(u, v)$. There are three cases to discuss.
	
	\begin{enumerate}[(1)]
		\item If $\mathit{path}(u, v)\cap M = \emptyset$, by Lemma \ref{struct-marked} we claim $\mathit{path}(u, v)$ has at most $2\log n$ hops, and then we can directly retrieve the answer of $Q(T(x), u, v)$ from precomputed entries of (\romannumeral2), each taking constant query time.
	
		\item Second, consider the case where $\mathit{path}(u, v)\cap M \neq \emptyset$. Let $s_1, s_2, \cdots, s_l, l\geq 1$ be the consecutive sequence (in ascending order with respect to depth in tree $T$) of all vertices from $M$ that are on $\mathit{path}(u, v)$. For those subtrees $T(x)$ that are hanging on $\mathit{path}(u, \mathit{par}(s_1))$, we can directly retrieve the value of $Q(T(x), u, \mathit{par}(x))$ from (\romannumeral2) in constant time, as by Lemma~\ref{struct-marked} $\mathit{path}(u, \mathit{par}(s_1))$ has at most $2\log n$ hops.

		\item Third, we turn to study the value of $Q(T(x), u, \mathit{par}(x))$ when $\mathit{par}(x)$ belongs to a $\mathit{path}(s_i, \mathit{par}(s_{i+1})), i<l$ or $\mathit{path}(s_l, v)$. The algorithm is two-fold.
		\begin{enumerate}[(a)]
			\item First, we make a query of $u$ to the fractional cascading data structure built at vertex $s_i$ ($1\leq i\leq l$), namely part (\romannumeral3), which would give us, for every descendent $y\in M_{s_i}$, the highest neighbour of $y$ below $u$. Using this information we are able to derive the result of $Q(T(x), u, v)$ if $|T(x)| < \log n$, since in this case $T(x)\cap M = \emptyset$ and thus $T(x)\subseteq M_{s_i}$.
			
			By Lemma \ref{fractional_cascading} the total time of this procedure is $O(|M_{s_i}| + \log n)$. 
			\item We are left to deal with cases where $|T(x)| \geq \log n$. In this case, we directly compute $Q(T(x), u, v)$ using the 2D-range successor built in (\romannumeral1) which takes $O(\log n)$ time.
		\end{enumerate}
	\end{enumerate}

	Correctness of the query algorithm is self-evident by the algorithm. The total query time is analysed as following. Throughout an execution of Algorithm \incremain, (1) and (2) contribute at most $O(n)$ time since each $T(x)$ is involved in at most one such query $Q(T(x), u, v)$ which takes constant time. As for (3)(a), since each marked vertex $s\in M$ lies in at most one such path $(w_1, w_2, \cdots, w_t) = \mathit{path}(u, v)$, the fractional cascading data structure associated with $M_s$ is queried for at most once. Hence the total time of (3)(a) is $O(\sum_{s\in M}(|M_s| + \log n)) = O(n + |M|\log n) = O(n)$; the last equality holds by $|M|\leq O(n / \log n)$ due to Lemma \ref{tree_partition_lem}.

	Finally we analyse the total time taken by (3)(b). It suffices to upper-bound by $O(n / \log n)$ the total number of such $x$ with the property that $|T(x)| \geq \log n$ and $\mathit{path}(u, \mathit{par}(x))\cap M \neq \emptyset$. Let $X$ be the set of all such $x$'s.
	\begin{lemma}\label{struct-large-tree}
		Suppose $x_1, x_2\in X$ and $x_1$ is an ancestor of $x_2$ in tree $T$. Then $\mathit{path}(x_1, x_2)\cap M \neq \emptyset$.
	\end{lemma}
	\begin{proof}
		Suppose otherwise $\mathit{path}(x_1, x_2)\cap M = \emptyset$. Consider the time when query $Q(T(x_2), u, v)$ is made and let $\mathit{path}(u, v)$ be the path being visited by then. As $x_2\in X$, by definition it must be $\mathit{path}(u, \mathit{par}(x_2))\cap M\neq \emptyset$. Therefore, $\mathit{path}(u, x_2)$ is a strict extension of $\mathit{path}(x_1, x_2)$, and thus $x_1, \mathit{par}(x_1)\in \mathit{path}(u, x_2)$, which means $x_1$ and $\mathit{par}(x_1)$ become visited in the same invocation of Algorithm \textsf{DFS}. This is a contradiction since for any query of form $Q(T(x_1), \cdot, \cdot)$ to be made, by then $\mathit{par}(x_1)$ should be tagged ``visited'' while $x_1$ is not.
	\end{proof}

	Now we prove $|X| = O(n / \log n)$. Build a tree $T[X]$ on vertices $X\cup \{r\}$ in the natural way: for each $x\in X$, let its parent in $T[X]$ be $x$'s nearest ancestor in $X\cup \{r\}$. Because of $$|X| < 2\#\text{leaves of }T[X] + \#\text{vertices with a unique child in }T[X]$$ it suffices to bound the two terms on the right-hand side: on the one hand, the number of leaves of $T[X]$ is at most $n / \log n$ since for each leave $x$ it has $|T(x)|\geq \log n$; on the other hand, for each $x\in T[X]$ with a unique child $y\in T[X]$, by Lemma \ref{struct-large-tree} $\mathit{path}(x, y)\cap M \neq \emptyset$, and so we can charge this $x$ to an arbitrary vertex in $\mathit{path}(x, y)\cap M$, which immediately bounds the total number of such $x$'s by $|M| = O(n / \log n)$; see Figure \ref{qry} for an illustration. Overall, $|X| \leq O(n / \log n)$.

\subparagraph*{Acknowledgments.}
The authors would like to thank Shahbaz Khan, Kasper Green Larsen and Seth Pettie for many helpful discussions, and the anonymous reviewer for pointing out an issue in an earlier version of this paper.  
\bibliography{biblio}
\end{document}